\documentclass[12pt]{article}
\usepackage[margin=1in]{geometry}
\usepackage{url}

\usepackage[T1]{fontenc}
\usepackage{fouriernc}

\usepackage{amsmath}
\usepackage{amssymb}
\usepackage{amsthm}
\usepackage{mathrsfs}

\def\shuffle{\sqcup\mathchoice{\mkern-7mu}{\mkern-7mu}{\mkern-3.2mu}{\mkern-3.8mu}\sqcup}

\newtheorem{theorem}{Theorem}
\newtheorem{corollary}[theorem]{Corollary}
\newtheorem{lemma}[theorem]{Lemma}

\theoremstyle{remark}
\newtheorem{remark}{Remark}
\newtheorem{example}{Example}

\begin{document}
\title{The Frobenius problem for the shuffle operation}

\author{Jeremy Nicholson\thanks{The author was supported by an NSERC
    USRA.}  { }and Narad Rampersad\thanks{The author was supported by
    NSERC Discovery Grant No.~418646--2012.} \\
Department of Mathematics and Statistics \\
University of Winnipeg \\
515 Portage Avenue \\
Winnipeg, Manitoba R3B 2E9 (Canada)\\
\url{jnich998@hotmail.com, n.rampersad@uwinnipeg.ca}}

\date{\today}
\maketitle

\begin{abstract}
Given a set $S$ of words, let $S^\dagger$ denote the \emph{iterated
  shuffle} of $S$.  We characterize the finite sets $S$ for which $S^\dagger$
is co-finite, and we give some bounds on the length of a longest word not in
$S^\dagger$.
\end{abstract}

\section{Introduction}
The classical \emph{Frobenius problem} is the following: Given
positive integers $m_1, m_2, \ldots, m_k$ such that
$\gcd(m_1,m_2,\ldots,m_k)=1$, what is the largest integer that cannot
be written as a non-negative integer linear combination of
$m_1,m_2,\ldots,m_k$?  Schur showed that this number always
exists and Sylvester showed that when $k=2$ this number is equal to
$m_1m_2-m_1-m_2$.  The case $k \geq 3$ is rather more difficult.  An
entire book has been devoted to this problem \cite{RA05}.

Shallit proposed the following ``non-commutative'' version of the
Frobenius problem: Given a set of words
$S = \{w_1, w_2, \ldots, w_k\}$ over an alphabet $\Sigma$ such that
$S^*$ is co-finite (i.e., contains all but finitely many words over
$\Sigma$), what is the length of a longest word not in $S^*$?  In
other words, what is the length of a longest word that cannot be
written as a \emph{concatenation} of a sequence of words chosen from
$S$?  Xu studied this problem, which he called the \emph{Frobenius
  problem in the free monoid}, in his Ph.D.\ thesis \cite{Xu09}.  Note
that this problem reduces to the classical Frobenius problem when
$\Sigma$ is a unary alphabet.  For larger alphabets, Xu considered the
special case where $S$ contains words of only two lengths $m$ and $n$
(say, $m<n$).  He showed that the answer to the problem in this
setting could be exponential in $n-m$.

In this paper we examine the same problem with respect to the
\emph{shuffle} operation on words.  Informally, the shuffle of two
words $u$ and $v$ is the set of all words that can be obtained by
``interleaving'' the letters of $u$ with the letters of $v$ in all
possible ways.  This shuffle operation on words was introduced, in an
algebraic setting, by Eilenberg and MacLane \cite{EM53}.  Some notable
early papers that study the shuffle operation from a formal languages
perspective are Jantzen \cite{Jan81, Jan85} and Warmuth and Haussler
\cite{WH84}.

Given a set $S$ of words, let $S^\dagger$ denote the \emph{iterated
  shuffle} of $S$ (see the formal definition in the next section).  In
this paper we characterize the finite sets $S$ for which $S^\dagger$
is co-finite, and we show that the length of a longest word not in
$S^\dagger$ is at most quadratic in the length of the longest word in
$S$.  The Frobenius problem in this setting therefore turns out to be
somewhat closer to the classical integer version of the problem,
rather than the ``free monoid'' version of the problem studied by Xu.

\section{Preliminaries}

Let us recall again the classical Frobenius problem:  Given
positive integers $m_1, m_2, \ldots, m_k$ such that
$\gcd(m_1,m_2,\ldots,m_k)=1$, what is the largest integer that cannot
be written as a non-negative integer linear combination of
$m_1,m_2,\ldots,m_k$?  Let $g(m_1,m_2,\ldots,m_k)$ denote this
quantity, which is known as the \emph{Frobenius number} for the given
instance of the problem.  For $k=2$ we have the exact formula of
Sylvester: $g(m_1,m_2) = m_1m_2 - m_1 - m_2$.  For $k\geq 2$ there are a
number of upper bounds; for our purposes, the following one due to
Schur (see \cite{Bra42}) will suffice: if $m_1 \leq m_2 \leq \cdots
\leq m_k$, then
\begin{equation}\label{Xu}
g(m_1,m_2,\ldots,m_k) \leq m_1m_k-m_1-m_k.
\end{equation}

Our goal in this paper is to generalize the Frobenius problem to the
setting of words over an alphabet.  Let $\Sigma$ denote an alphabet
and let $\Sigma^*$ denote the set of all words over $\Sigma$.  For any
$a \in \Sigma$ and $w \in \Sigma^*$, the number of occurrences of $a$
in $w$ is denoted by $|w|_a$.

Let $|\Sigma| = q$ and let $S$ be a finite subset of $\Sigma^*$ such
that $S^*$ is co-finite.  Xu showed that if $q>1$ and $S$ contains
words of lengths $m$ and $n$ only $(m<n)$, then the longest word not
in $S^*$ has length at most $mq^{n-m} + n - m$ and this bound is
tight.

The Kleene star operator used above can be viewed as iterated
concatenation.  In this paper we will study the Frobenius problem for
iterated shuffle.  The \emph{shuffle operator} can be defined
as follows
\[
u \shuffle v = \{u_1v_1u_2v_2\cdots u_kv_k\; :\; u_i, v_i \in \Sigma^*,\;
 u = u_1u_2\cdots u_k, \text{ and } v=v_1v_2\cdots v_k\}.
\]
The \emph{iterated shuffle}\footnote{There is no standard notation to
  denote the iterated shuffle; we are following Jantzen's use of the
  ``dagger''.} of a word is defined by
\[
u^\dagger = \bigcup_{i=0}^\infty \underbrace{(u \shuffle u \shuffle
  \cdots \shuffle u)}_{i\text{ times}}.
\]
We extend both of these operations from words to sets of words in the
usual way.  First, for sets of words $A$ and $B$ we define
\[
A \shuffle B = \{ u \shuffle v \;:\; u \in A \text{ and } v \in B\},
\]
and
\[
A^\dagger = \bigcup_{i=0}^\infty \underbrace{(A \shuffle A \shuffle
  \cdots \shuffle A)}_{i\text{ times}}.
\]
 For example,
\[ \{01,011\} \shuffle \{0\} = \{001, 010, 0011, 0101, 0110\}\]
and
\[ \{01,2\}^\dagger = \{w \in \{0,1,2\}^* : |w|_0 = |w|_1 \text{ and
  every prefix $u$ of $w$ satisfies } |u|_0 \geq |u|_1\}. \]
Over the alphabet $\{0,1\}$ the set $\{00,000,11,111,01,10\}^\dagger$
is co-finite; the only strings not in this set are
$0, 1, 001, 010, 100, 011, 101, 110.$

The \emph{Frobenius problem for the shuffle operation} is therefore
the following: If $S$ is a finite set of words over $\Sigma$ such that
$S^\dagger$ is co-finite, what is the length of a longest word not in
$S^\dagger$?  Note that when $|\Sigma|=1$, say $\Sigma=\{x\}$, this
problem is equivalent to the integer Frobenius problem.  In this case,
either $x \in S$, in which case $S^\dagger = \Sigma^*$, or $S$ has the
form $S = \{x^{m_1},x^{m_2},\ldots,x^{m_k}\}$, where $m_i \geq 2$ for
$i=1,2,\ldots,k$.  The length of any word in $S^\dagger$ is a multiple
of $\gcd(m_1,m_2,\ldots,m_k)$, so for $S$ to be co-finite we must have
$\gcd(m_1,m_2,\ldots,m_k)=1$.  In this case it is clear that
$x^n \in S^\dagger$ if and only if $n$ can be written as a
non-negative integer linear combination of $m_1,m_2,\ldots,m_k$.  Hence,
the length of the longest word not in $S^\dagger$ is exactly
$g(m_1,m_2,\ldots,m_k)$.

Lastly we define the act of \emph{matching} a word. Given two words
$w = w_1 \cdots w_n$ and $u = u_1 \cdots u_k$, a \emph{match} of $u$
in $w$ is a subset of positions $i_1 < i_2 < \ldots < i_k$ such that
$w_{i_j} = u_j$ for $j = 1,2,\ldots, k$.  For a given set of words
$S$, a word $y$ is in $S^{\dagger}$ if and only if all positions of
$y$ can be covered by a pairwise disjoint set of matches using words
in $S$. In this case, we say $y$ can be \emph{matched} using words in
$S$.  For example, if $S=\{011, 012\}$, then $010121\in S^{\dagger}$
since it can be matched using a single $011$ and a single $012$
($\underline{01}\overline{012}\underline{1}$ or
$\overline{01}\underline{01}\overline{2}\underline{1}$). However,
$01012112\not\in S^{\dagger}$ because one of the $0$'s will have to be
matched twice to ensure the last $1$ and $2$ get matched.

\section{A characterization of $S$ such that $S^\dagger$ is co-finite}

In this section we give a complete characterization of the finite sets $S$ for which $S^\dagger$
is co-finite and we give some bounds on the length of a longest word not in $S^\dagger$. Furthermore, given a set $S$ such that $S^{\dagger}$ is co-finite, we give a detailed description on how to match a sufficiently long arbitrary word using the words in $S$ and we give lower bounds for $|S|$.

For the rest of this section, let the following be true. 
Let $q \geq 1$ and let $S$ be a finite set of words over an alphabet
$\{x_1,x_2,\ldots,x_q\}$.  For each $i \in \{1,2,\ldots,q\}$, let
$\mathscr{T}_i$ denote the collection of all subsets $T_i \subseteq S$,
where either
\begin{equation}\label{smallTi} T_i = \{x_i\},\end{equation} or
\begin{equation}\label{bigTi} T_i=\{x_i^{m_{i, 1}}, x_i^{m_{i, 2}}, \ldots,
x_i^{m_{i, h_i}}\} \cup \bigcup_{\substack{j=1 \\ j\neq i}}^{q}
\{x_ix_j^{a_{i, j}}, x_j^{b_{j, i}}x_i\},\end{equation}
where $m_{i, h_i}> \cdots > m_{i, 2}>m_{i, 1}\geq 2$, $h_i\geq2$,
$\gcd(m_{i, 1}, m_{i, 2}, \ldots , m_{i, h_i})=1$ and
$a_{i, j}, b_{j, i}\geq 1$ for all $j\in \{1, 2, \ldots, q\}$ such
that $i\neq j$.

Our first main result is the following.

\begin{theorem}
\label{cofinite_characterization}
The set $S^\dagger$ is co-finite if and only if for each $i \in\{1,2,\ldots,q\}$, the collection $\mathscr{T}_i$ is non-empty.
\end{theorem}

\begin{proof}

Follows directly from Lemma~\ref{cofinite_only_if} and
Theorem~\ref{upper bound} below.

\end{proof}


\begin{lemma}\label{cofinite_only_if}

The set $S^\dagger$ is co-finite only if for each $i \in
\{1,2,\ldots,q\}$, the collection $\mathscr{T}_i$ is non-empty.

\end{lemma}

\begin{proof}

Let $x_i$ be an arbitrary letter in our alphabet and assume $\{x_i\}\not\subseteq S$. If $S^{\dagger}$ is co-finite, then every sufficiently long string of $x_i$'s must be in $S^{\dagger}$. So $S$ must have a string of $x_i$'s, notably $x_i^{m_{i, 1}}$. However, then $S^{\dagger}$ will only contain strings of $x_i$'s that are multiples of $m_{i, 1}$ in length. Thus $S$ must have at least two strings of $x_i$'s, notably $x_i^{m_{i, 1}}$, $x_i^{m_{i, 2}}$, \ldots, $x_i^{m_{i, h_i}}$ (where $m_{i, h_i}> ...> m_{i, 2}>m_{i, 1}\geq 2$ and $h_i\geq2$). Furthermore, every string of $x_i$'s in $S^{\dagger}$ is a multiple of $\gcd(m_{i, 1}, m_{i, 2}, \ldots, m_{i, h_i})$ in length, so if $S^{\dagger}$ is co-finite, then $\gcd(m_{i, 1}, m_{i, 2}, \ldots, m_{i, h_i})=1$. Let $x_j$ be another arbitrary letter in our alphabet. Then $x_ix_j^s$ and $x_j^sx_i$ are in $S^{\dagger}$ for all sufficiently large $s$. Since $\{x_i\}\not\subseteq S$, 
$$\bigcup_{\substack{j=1 \\ j\neq i}}^{q}\{x_ix_j^{a_{i, j}}, x_j^{b_{j, i}}x_i\}$$ must be a subset of $S$ where $a_{i, j}, b_{j, i}\geq 1$ for all $j\in \{1, 2, \ldots, q\}$ such that $i\neq j$. Thus $\mathscr{T}_i$ is non-empty and the result follows.

\end{proof}

\begin{theorem}\label{upper bound}

Assume for each $i \in\{1,2,\ldots,q\}$, the collection $\mathscr{T}_i$ is non-empty. 
If $\{x_i\}\in\mathscr{T}_i$, define $g_i=-1$ and $m_{i, 1}=a_{i, j}=b_{j, i}=0$ for all $j\in\{1,2,\ldots,q\}$. 
If $\{x_i\}\notin\mathscr{T}_i$, then let $T_i$ be a subset of $S$ of the form \eqref{bigTi}. Define $g_i=g(m_{i, 1}, m_{i, 2}, ..., m_{i, h_i})$ and let $m_{i, 1}$, $m_{i, 2}$, \ldots, $m_{i, h_i}$, $a_{i, j}$, $b_{j, i}$ be as they are defined for the chosen $T_i$ for all $j\in\{1,2,\ldots,q\}$.
Then $S^\dagger$ contains every word of length at least
$$\sum_{i=1}^q g_i +q+(q-1)(\max_i(m_{i, 1}))(\max_{i, j,  j\neq i} (a_{i, j}, b_{j, i})).$$

\end{theorem}

\begin{proof}

Let $y$ be a $q$-ary word of length at least $$\sum_{i=1}^q g_i +q+(q-1)(\max_i (m_{i, 1}))(\max_{i, j, j\neq i} (a_{i, j}, b_{j, i}))$$.
If $|y|_{x_i}\geq g_i+1$ for all $i$, every $x_i$ in $y$ can be matched by shuffling the set $\{x_i\}$ or if $\{x_i\}\not\subseteq S$, by shuffling the set $\{x_i^{m_{i, 1}}, x_i^{m_{i, 2}}, \ldots, x_i^{m_{i, h_i}}\}$. Furthermore, if $|y|_{x_i}\leq g_i$ for all $i$, then $|y|\leq \sum_{i=1}^q g_i$, a contradiction.
So assume WLOG that $|y|_{x_1}\leq g_1$, $|y|_{x_2}\leq g_2$, \ldots, $|y|_{x_s}\leq g_s$, $|y|_{x_{s+1}}>g_{s+1}$, $|y|_{x_{s+2}}>g_{s+2}$, \ldots, $|y|_{x_q}>g_q$ for some $s$ such that $1\leq s\leq q-1$. 
So $$\sum_{i=1}^s |y|_{x_i}\leq \sum_{i=1}^s g_i$$ which implies
\begin{eqnarray*}
\sum_{i=s+1}^q |y|_{x_i} &\geq& 
\sum_{i=s+1}^q g_i +q+(q-1)(\max_i (m_{i, 1}))(\max_{i, j, j\neq i} (a_{i, j}, b_{j, i}))\\
&\geq& \sum_{i=s+1}^q g_i +q-s+(q-1)(\max_i (m_{i, 1}))(\max_{i, j, j\neq i} (a_{i, j}, b_{j, i})).
\end{eqnarray*}

For simplicity, let $\lambda=(\max\limits_i (m_{i,
  1}))(\max\limits_{i, j, j\neq i} (a_{i, j}, b_{j, i}))$ and assume
for all $i$ such that $s+1\leq i\leq q$, we have $$|y|_{x_i}=g_i+1+\gamma_i\lambda + r_i,$$ where $\gamma_i$ is a non-negative integer and $0\leq r_i< \lambda$.
It follows that 
\begin{eqnarray*}
\sum_{i=s+1}^q |y|_{x_i} &=& \sum_{i=s+1}^q g_i +q-s+\sum_{i=s+1}^{q} \gamma_i\lambda+\sum_{i=s+1}^{q} r_i\\
&\geq& \sum_{i=s+1}^q g_i +q-s+(q-1)\lambda,
\end{eqnarray*}
which implies $$\sum_{i=s+1}^q \gamma_i\lambda + \sum_{i=s+1}^q r_i \geq (q-1)\lambda=(s-1)\lambda+(q-s)\lambda.$$
Since $r_i<\lambda$ for all $i$ between $s+1$ and $q$, we have
$$\sum_{i=s+1}^q r_i<(q-s)\lambda.$$
So we get that
\begin{eqnarray*}
\sum_{i=s+1}^q\gamma_i\lambda &\geq&
                                     (s-1)\lambda+(q-s)\lambda-\sum_{i=s+1}^{q} r_i\\
&>&(s-1)\lambda+(q-s)\lambda-(q-s)\lambda\\
&=&(s-1)\lambda.
\end{eqnarray*}
Since $\sum_{i=s+1}^{q} \gamma_i$ is a non-negative integer, the
previous inequality $\sum_{i=s+1}^{q}\gamma_i\lambda>(s-1)\lambda$
implies that
\begin{equation}\label{gamma summation}
\sum_{i=s+1}^{q}\gamma_i\geq s.
\end{equation}

(When reading the remainder of this proof the reader may wish to refer
to Example~\ref{matching_ex} below.)

For simplicity, let $C=\{x_1, x_2, \ldots, x_s\}$ and $D=\{x_{s+1},
x_{s+2}, \ldots, x_q\}$. Note that if $x_i\in C$, then there exists a
$T_i$ of the form $\eqref{bigTi}$ such that $T_i\subseteq S$. (If
$x_i\in S$, then $g_i = -1$ and so $|y|_{x_i}\geq0$ implies $x_i\in
D$.) If $s<\gamma_{s+1}$, then associate all $s$ letters in $C$ with
$x_{s+1}$; otherwise, associate the first $\gamma_{s+1}$ letters from
$C$ with $x_{s+1}$. If $s<\gamma_{s+1}+\gamma_{s+2}$, then associate
the remaining $s-\gamma_{s+1}$ letters from $C$ with $x_{s+2}$;
otherwise, associate the next $\gamma_{s+2}$ letters from $C$ with
$x_{s+2}$. If $s<\gamma_{s+1}+\gamma_{s+2}+\gamma_{s+3}$, then
associate the remaining $s-\gamma_{s+1}-\gamma_{s+2}$ letters from $C$
with $x_{s+3}$; otherwise, associate the next $\gamma_{s+3}$ letters from $C$ with $x_{s+3}$. Repeat this process until every letter in $C$ has an associated letter in $D$ (which we know is possible by \eqref{gamma summation}). Once this process is completed, every letter in $C$ will be associated with exactly one letter in $D$ and every letter $x_i$ in $D$ will be associated with at most $\gamma_i$ letters of $C$. Let $x_c$ be a letter from $C$ and let $x_d$ be its associated letter in $D$ (note that $\gamma_d\geq1$ by definition). By the division algorithm, we get $|y|_{x_c}=q_cm_{c, 1}+r_c$ where $0\leq r_c< m_{c, 1}$. Consider the first occurrence of a $x_c$ in $y$. If it is preceded by $b_{d, c}$ $x_d$'s, then use a $x_d^{b_{d, c}}x_c$ to match the first $x_c$ and the first $b_{d, c}$ $x_d$'s. If not, then it must be followed by $a_{c, d}$ $x_d$'s. If that is the case, use a $x_cx_d^{a_{c, d}}$ to match the first $x_c$ and the first $a_{c, d}$ $x_d$'s that follow the first $x_c$. Now consider the second occurrence of a $x_c$. If it is preceded by $b_{d, c}$ unmatched $x_d$'s, then use a $x_d^{b_{d, c}}x_c$ to match the second $x_c$ and the first previously unmatched $b_{d, c}$ $x_d$'s. If not, then it must be followed by $a_{c, d}$ unmatched  $x_d$'s. If this is the case, then use a $x_cx_d^{a_{c, d}}$ to match the second $x_c$ and the first $a_{c, d}$ unmatched $x_d$'s that follow the second $x_c$. Repeat this process for the first $r_c$ $x_c$'s. This process will work because once we get to the $r_c$th $x_c$, we still have at least
\begin{eqnarray*} 
&\phantom{=}&(\max_i (m_{i, 1}))(\max_{i, j, j\neq i}(a_{i, j}, b_{j, i}))-(r_c-1)(\max(a_{c, d}, b_{d, c}))\\
&\geq&(\max_i(m_{i, 1}))(\max_{i, j, j\neq i}(a_{i, j}, b_{j, i}))-(m_{c, 1}-1-1)(\max(a_{c, d}, b_{d, c}))\\
&\geq&(\max_i(m_{i, 1}))(\max_{i, j, j\neq i}(a_{i, j}, b_{j, i}))-(\max_i(m_{i, 1})-2)(\max_{i, j, j\neq i}(a_{i, j}, b_{j, i}))\\
&=&2(\max_{i, j, j\neq i}(a_{i, j}, b_{j, i}))\\
&\geq& 2(\max(a_{c, d}, b_{d, c}))
\end{eqnarray*}
unmatched $x_d$'s. This ensures that the $r_c$th $x_c$ is either preceeded by $b_{d, c}$ $x_d$'s or followed by $a_{c, d}$ $x_d$'s. Once we match the first $r_c$ $x_c$'s, use $q_c$ $x_c^{m_{c, 1}}$'s to match the remaining $x_c$'s. This procedure will match every $x_c$ in $y$ and less than $\lambda$ $x_d$'s. Repeat the procedure above for every letter in $C$ with respect to their associated letter in $D$ so that every occurrence of a letter in $C$ in $y$ has been matched. Let $x_i$ be an arbitrary letter in $D$. After every letter in $C$ has been matched in $y$, less than $\gamma_i\lambda$ $x_i$'s have been matched. Thus, for each $i$ such that $s+1\leq i\leq q$, there are at least $g_i+1$ unmatched $x_i$'s in $y$. In which case use $x_i$'s or if $x_i\not\in S$, use the elements in the set $$\{x_i^{m_{i, 1}}, x_i^{m_{i, 2}}, \ldots, x_i^{m_{i, h_i}}\}$$ to match the remaining $x_i$'s in $y$. Every digit in $y$ has now been matched and therefore, $y\in S^{\dagger}$.

\end{proof}

\begin{remark}

If $S^{\dagger}$ is co-finite, then for all $i$, there exists a $T_i$ in $S$ of the form \eqref{smallTi} or \eqref{bigTi}. 
It should be noted that for any $i$, several $T_i$'s of the form
\eqref{bigTi} could be in $S$ and if $\{x_i\}\not\subseteq S$, our
choice of $T_i$ could impact the bound we obtain in Theorem~\ref{upper
  bound}. 
Let $i$ be arbitrary and assume $\{x_i\}\not\subseteq S$. 
To obtain the smallest possible bound in Theorem~\ref{upper bound}, choose a $T_i$ of the form \eqref{bigTi} that includes every string of $x_i$'s in $S$ (ensuring $g_i$ and $m_{i, 1}$ are as small as possible) and strings of the form $x_ix_j^r$ and $x_j^sx_i$ for each $j$ where $r$ and $s$ are as small as possible (ensuring $a_{i, j}$ and $b_{j, i}$ are as small as possible for all $j$).

\end{remark}

The following example illustrates the procedure described in the proof
of Theorem~\ref{upper bound}.

\begin{example}\label{matching_ex}
Consider the set 
\begin{eqnarray*}
S=\{x_1^2, x_1^5, x_2^3, x_2^4, x_3^2, x_3^3, x_4, x_1x_2,x_2^2x_1,x_2x_1^3, x_1x_3, x_3x_1, \\
x_1x_4^2, x_4^3x_1, x_2x_3^2, x_3^2x_2, x_3x_2^3, x_2^2x_3, x_2x_4^3, x_4^3x_2, x_3x_4^3, x_4^2x_3\}. 
\end{eqnarray*}
It can be observed that the set $S$ is of the required construction to make $S^{\dagger}$ co-finite by Theorem~\ref{upper bound}. Furthermore, we get the following values:
\begin{eqnarray*}
&g_1&=(2)(5)-2-5=3,\\
&g_2&=(3)(4)-3-4=5,\\
&g_3&=(2)(3)-2-3=1,\\
&g_4&=-1,\\
&\max\limits_i (m_{i, 1})&=\max(2, 3, 2, 0)=3,\\
&\max\limits_{i, j, j\neq i} (a_{i, j}, b_{i, j})&=\max(1, 2, 3, 1, 1, 1, 1, 1, 2, 3, 0, 0, 2, 2, 3, 2, 3, 3, 0, 0, 3, 2, 0, 0)=3.
\end{eqnarray*}
By Theorem~\ref{upper bound}, $S^{\dagger}$ contains every $4$-ary word of length at least 
$$3+5+1+(-1)+4+(4-1)(3)(3)=39.$$ Consider the word 
$$y=x_3x_3x_4x_2x_4x_3x_3x_4x_1x_4x_2x_3x_4x_3x_4x_2x_4x_1x_3x_4x_4x_1
x_4x_4x_2x_3x_4x_4x_4x_3x_3x_4x_3x_3x_4x_4x_3x_4x_2.$$ 
Since $|y|=39$, we will be able to match all the letters in $y$ using the procedure detailed in Theorem~\ref{upper bound} above. For simplicity, let $i$ denote $x_i$ for all $i$ such that $1\leq i\leq 4$. So we get 
$$y=334243341423434241344144234443343344342.$$

Step 1: Determine which letters go in the sets $C$ and $D$.

Since $|y|_1=3\leq g_1$ and $|y|_2=5\leq g_2$, $1, 2\in C$. Since $|y|_3=13>g_3$ and $|y|_4=18>g_4$, $3, 4\in D$.

Step 2: Determine the associated $\gamma$ for each element of $D$.

Note that $\lambda=9$. Then $$|y|_3=13=1+1+(1)(9)+2=g_3+1+\gamma_3\lambda+r_3\implies \gamma_3=1$$ and $$|y|_4=18=-1+1+(2)(9)+0=g_4+1+\gamma_4\lambda+r_4\implies \gamma_4=2.$$

Step 3: Associate every letter in $C$ with a single letter in $D$.

Since $\gamma_3=1<2=|C|$, associate only the first letter in $C$ with 3 in $D$. 
Thus, 1 is associated with 3. 
Since $|C|\leq\gamma_3+\gamma_4$, associate the remaining letter(s) in $C$ with 4. 
Thus, 2 is associated with 4.

Step 4: Match every occurrence of a letter in $C$ in $y$.

We will match the 1's first. Note that $|y|_1=(1)(2)+1=q_1m_{1,
  1}+r_1$. In $y$, the first occurrence of a 1 is preceded by $b_{3,
  1}=1$ occurrences of 3. Thus, use a 31 to match the first 3 and the first 1. Note that in all the lines below, a letter that is underlined is currently being matched and a letter that is overlined has been previously matched.
$$31\implies \underline{3}3424334\underline{1}423434241344144234443343344342$$
We have matched $r_1$ 1's, so now we need $q_1=1$ copies of $1^{m_{1, 1}}=1^2$ to match the remaining 1's.
$$11\implies \overline{3}3424334\overline{1}42343424\underline{1}344\underline{1}44234443343344342$$
Now we match the 2's. Note that $|y|_2=(1)(3)+2=q_2m_{2, 1}+r_2$. The
first occurrence of a 2 in $y$ is not preceded by $b_{4, 2}=3$
occurrences of 4. Thus, we use a $24^{a_{2, 4}}=24^3$ to match the first 2 and the first three 4's that follow the first 2.
$$2444\implies \overline{3}34\underline{24}33\underline{4}\overline{1}\underline{4}2343424\overline{1}344\overline{1}44234443343344342$$
The second occurrence of a 2 in $y$ is not preceded by three unmatched 4's. Thus, we will use another $24^3$ to match the second 2 and the first three unmatched 4's that follow the second 2.
$$2444\implies \overline{3}34\overline{24}33\overline{414}\underline{2}3\underline{4}3\underline{4}2\underline{4}\overline{1}344\overline{1}44234443343344342$$
We have matched $r_2$ 2's, so now we need $q_2=1$ copies of $2^{m_{2, 1}}=2^3$ to match the remaining 2's.
$$222\implies \overline{3}34\overline{24}33\overline{4142}3\overline{4}3\overline{4}\underline{2}\overline{41}344\overline{1}44\underline{2}3444334334434\underline{2}$$

Step 5: Match every remaining occurrence of a letter in $D$ in $y$.

First we match the remaining 3's. Since there are at least $g_3+1=2$ unmatched 3's remaining, we can use the $3^{m_{3, 1}}=3^2$ and $3^{m_{3, 2}}=3^3$ strings to match the remaining 3's. Since there are twelve unmatched 3's remaining, we can use six $3^2$'s.
$$33\implies \overline{3}\underline{3}4\overline{24}\underline{3}3\overline{4142}3\overline{4}3\overline{4}\overline{241}344\overline{1}44\overline{2}3444334334434\overline{2}$$
$$33\implies \overline{33}4\overline{243}\underline{3}\overline{4142}\underline{3}\overline{4}3\overline{4}\overline{241}344\overline{1}44\overline{2}3444334334434\overline{2}$$
$$33\implies \overline{33}4\overline{2433414234}\underline{3}\overline{4}\overline{241}\underline{3}44\overline{1}44\overline{2}3444334334434\overline{2}$$
$$33\implies \overline{33}4\overline{2433414234342413}44\overline{1}44\overline{2}\underline{3}444\underline{3}34334434\overline{2}$$
$$33\implies \overline{33}4\overline{2433414234342413}44\overline{1}44\overline{23}444\overline{3}\underline{3}4\underline{3}34434\overline{2}$$
$$33\implies \overline{33}4\overline{2433414234342413}44\overline{1}44\overline{23}444\overline{33}4\overline{3}\underline{3}44\underline{3}4\overline{2}$$
Next we match the remaining 4's. Since $4\in S$, clearly we can match the twelve remaining unmatched 4's using twelve 4's (for simplicity, the line below symbolizes repeating the process twelve times).
$$4\implies \overline{33}\underline{4}\overline{2433414234342413}\underline{44}\overline{1}\underline{44}\overline{23}\underline{444}\overline{33}\underline{4}\overline{3}\overline{3}\underline{44}\overline{3}\underline{4}\overline{2}$$
Now we have matched every letter in $y$ and it follows that $y\in S^{\dagger}$.
\end{example}

\begin{corollary}
Assume $S^{\dagger}$ is co-finite and let $n$ be the length of a longest word in the set $S$. 
Then the length of a longest word not in $S^{\dagger}$ is less or equal to $(2q-1)n^2-(5q-2)n+3q-2$.
\end{corollary}

\begin{proof}

For simplicity, let $k$ denote the length of a longest word not in $S^{\dagger}$. By Lemma~\ref{cofinite_only_if}, for each $i \in\{1,2,\ldots,q\}$, the collection $\mathscr{T}_i$ is non-empty. By Theorem~\ref{upper bound} (if $\{x_i\}\not\subseteq\mathscr{T}_i$ for some $i$, then let $T_i$ be a subset of $S$ of the form \eqref{bigTi}), $k$ is at most $$\sum_{i=1}^q g_i +q+(q-1)(\max_i(m_{i, 1}))(\max_{i, j,  j\neq i} (a_{i, j}, b_{j, i}))-1.$$
It is clear that $\max\limits_i (m_{i, 1})\leq n-1$ and $\max\limits_{i, j, j\neq i} (a_{i, j}, b_{j, i})\leq n-1$. 
Each $g_i$ is either $-1$ or $$g(m_{i, 1}, m_{i, 2}, \ldots, m_{i, h_i})\leq(m_{i, 1})(m_{i, h_i})-m_{i, 1}-m_{i, h_i}\leq (n-1)n-(n-1)-n=n^2-3n+1$$ 
by \eqref{Xu}. Since $n^2-3n+1\geq-1$ for all $n\in\mathbb{Z}$, it follows that 
\begin{eqnarray*}
k&\leq&\sum_{i=1}^q g_i +q+(q-1)(\max_i (m_{i, 1}))(\max_{i, j, j\neq i} (a_{i, j}, b_{j, i}))-1\\
&\leq&q(n^2-3n+1)+q+(q-1)(n-1)^2-1\\
&=&qn^2-3qn+q+q+qn^2-n^2-2qn+2n+q-2\\
&=&(2q-1)n^2-(5q-2)n+3q-2.
\end{eqnarray*}

\end{proof}

\begin{remark}
The last corollary implies that the length of a longest word not in $S^{\dagger}$ is at most quadratic in the length of a longest word in $S$.
\end{remark}

\begin{theorem}
Assume $S^\dagger$ is co-finite and let $m$
be the length of a shortest word in $S$.
\begin{itemize}
\item If $m = 1$ then $|S| \geq q$.
\item If $m = 2$ then $|S| \geq q^2 + q$.
\item If $m \geq 3$ then $|S| \geq 2q^2$.
\end{itemize}
In each case the bound on $|S|$ is tight.
\end{theorem}

\begin{proof}

It is clear that $S$ must contain at least one word with just $x_i$'s for every $i$. It follows that $|S|\geq q$. The set $S=\{x_1, x_2, \ldots, x_q\}$ is a set of size $q$ such that $S^{\dagger}$ is co-finite. 
Assume $m\geq2$. Then $\{x_i\}\not\subseteq S$ for all $i$. To ensure every sufficiently long string of $x_i$'s is in $S^{\dagger}$ for all $i$, $S$ must have at least two strings of $x_i$'s for all $i$. Furthermore, to ensure $x_ix_j^s\in S^{\dagger}$ for all $i$, $j$ and sufficiently large $s$ such that $i\neq j$, $S$ must contain a string consisting of a single $x_i$ followed by at least one $x_j$ for all $i$ and $j$ such that $i\neq j$. Therefore, $|S|\geq 2q+q(q-1)=q^2+q$. It can be observed (by Theorem~\ref{upper bound}) that the set $$S=\bigcup_{i=1}^q\{x_i^2, x_i^3\} \cup \bigcup_{\substack{j=1 \\ j\neq i}}^{q}\bigcup_{i=1}^q\{x_ix_j\}$$ is a set of size $q^2+q$ such that $S^{\dagger}$ is co-finite. 
Now assume $m\geq3$. By Lemma~\ref{cofinite_only_if}, for all $i \in
\{1,2,\ldots,q\}$, there exists a $T_i$ in $S$ of the form \eqref{bigTi}. Furthermore, all of the $T_i$'s are pairwise disjoint (i.e. if $i\neq j$, then $T_i\cap T_j=\emptyset$). Thus, $|S|\geq q|T_i|\geq q(2+2(q-1))=2q^2$. 
It can be observed (by Theorem~\ref{upper bound}) that the set $$S=\bigcup_{i=1}^q\{x_i^m, x_i^{m+1}\} \cup \bigcup_{\substack{j=1 \\ j\neq i}}^{q}
\bigcup_{i=1}^q\{x_ix_j^{m-1}, x_j^{m-1}x_i\}$$ is a set of size $2q^2$ such that $S^{\dagger}$ is co-finite.

\end{proof}

\begin{remark} Note that the lower bound on $|S|$ does not depend on
  $m$.  This is significantly different from the case of
  concatenation: for $S^*$ to be co-finite, the cardinality of $S$
  must be exponential in $m$.
\end{remark}

\section{A prototypical case}

To find a general formula for the length of a longest word not in
$S^{\dagger}$ for an arbitrary set $S$ such that $S^{\dagger}$ is
co-finite is a difficult task. So instead we restrict our efforts to a
family of what we deem to be the simplest constructions of $S$ such
that $S^{\dagger}$ is co-finite. We define each $T_i$ in the following
way: $$T_i=\{x_i^m, x_i^{m+1}\} \cup \bigcup_{\substack{j=1 \\ j\neq
    i}}^{q}\{x_ix_j^{m-1}, x_j^{m-1}x_i\}.$$ This results in a set $S$
of the form $$S=\bigcup_{i=1}^{q} \{x_i^m, x_i^{m+1}\} \cup \bigcup_{\substack{j=1 \\ j\neq i}}^{q} \bigcup_{i=1}^{q} \{x_ix_j^{m-1}, x_j^{m-1}x_i\}.$$ The cases when $q=1$ and $m=1$ are both trivial so we restrict our attention to when $q, m\geq2$. In this section we prove that the length of the longest word not in $S^{\dagger}$ is $2q-1$ when $m=2$ and $qm^2-2qm+2m-1$ when $m\geq 3$. We also find some elementary bounds on the number of words not in $S^{\dagger}$. For the rest of this section, unless explicitly stated, assume the set $S$ has the construction above.

\begin{theorem}

If $q\geq1$ and $m=2$, the length of a longest word not in $S^{\dagger}$ is $2q-1$.

\end{theorem}

\begin{proof}

Consider the word $y=x_1^2x_2^2\cdots x_{q-1}^2x_q$. It is clear that
we cannot use a $x_i^3$ to match any letters in $y$ for any $i$. The
only other words in $S$ are of even length. Since it is impossible to
match a word of odd length with only even words, $y\not\in
S^{\dagger}$. We claim that every $q$-ary word of length at least $2q$
is in $S^{\dagger}$. When $q=1$, clearly every unary word of length at
least 2 is in $S^{\dagger}$ (since $g(2, 3)=1$). Assume the claim
holds for some $q\geq1$. It suffices to show that the claim holds for
$q+1$. Let $y$ be a $(q+1)$-ary word of length at least $2(q+1)=2q+2$. If every letter has either 0 or at least 2 occurrences, then $y$ can be matched using the set $$\bigcup_{i=1}^{q+1} \{x_i^2, x_i^3\}.$$ If not, then there exists an $i$ such that $|y|_{x_i}=1$. Let $x_j$ be another letter in $y$. Then we can either use an $x_ix_j$ or an $x_jx_i$ to match the only $x_i$ and one of the $x_j$'s. Thus, $y$ is a $q$-ary word of length at least $2q$ (which we know is in $S^{\dagger}$ by our induction hypothesis) shuffled with a word from $S$. Thus $y\in S^{\dagger}$ and the result follows by induction on $q$.
\end{proof}

\begin{theorem}\label{simple_case_bound}

If $q\geq 2$ and $m\geq 3$, the length of a longest word not in $S^{\dagger}$ is $qm^2-2qm+2m-1$.

\end{theorem}

\begin{proof}

Follows directly from Lemmas~\ref{simple_case_existence},
\ref{simple_case_implication} and \ref{simple_case_every_word} below.

\end{proof}

\begin{lemma}\label{simple_case_existence}

If $q\geq 2$ and $m\geq 3$, there exists a $q$-ary word $y$ of length $qm^2-2qm+2m-1$ such that $y\not\in S^{\dagger}$.

\end{lemma}

\begin{proof}

Let $y=x_1^{m-2}x_2^{m-2}\cdots x_{q-1}^{m-2}x_q^{qm^2-3qm+3m+2q-4}x_1$ and assume for the sake of contradiction that $y\in S^{\dagger}$. By observation, the only way we can match the $x_1$'s is if we use $m-2$ $x_1x_q^{m-1}$'s and a single $x_q^{m-1}x_1$. Furthermore, the only way we can match the $x_2$'s is using $m-2$ $x_2x_q^{m-1}$'s. We repeat this process until we match the $x_{q-1}$'s with $m-2$ $x_{q-1}x_q^{m-1}$'s. This leaves $qm^2-3qm+3m+2q-4-(q-1)(m-2)(m-1)-(m-1)=m^2-m-1$ unmatched $x_q$'s and since $g(m, m+1)=m^2-m-1$, it is impossible to match the remaining unmatched $x_q$'s. It follows that $y\not\in S^{\dagger}$.

\end{proof}

\begin{lemma}\label{enough_of_a_letter}

If $q\geq1$ and $m\geq3$, then every $q$-ary word $y$ of length at least
$(q-1)m^2-2qm+3m+1$ has at least one letter $x_i$ such that $|y|_{x_i}\geq m+1$.

\end{lemma}

\begin{proof}

It is clear that a longest $q$-ary word that does not contain $m+1$ occurrences of any letter is of length $qm$.
The result follows if $(q-1)m^2-2qm+3m+1\geq qm+1$. We verify that

$$(q-1)m^2-2qm+3m+1-(qm+1)=(q-1)m^2-3(q-1)m=(m-3)(q-1)m\geq0.$$

\end{proof}

\begin{lemma}\label{matching_multiples_of_m}

If $q\geq1$ and $m\geq3$, then every $q$-ary word $y$ whose length is
a multiple of $m$ and is at least $(q-1)m^2-2qm+3m$ satisfies $y = w
\shuffle z$, where $w$ is a word in $S^{\dagger}$ and $z$ has length
exactly $(q-1)m^2-2qm+3m$.

\end{lemma}

\begin{proof}

If $|y| > (q-1)m^2-2qm+3m$ then by Lemma~\ref{enough_of_a_letter}
there is some letter $x_i$ such that $|y|_{x_i} \geq m$.  We can
therefore match $x_i^m$ in $y$.  Repeat this process until there are
exactly $(q-1)m^2-2qm+3m$ unmatched letters.  The matched letters form
a word $w \in S^\dagger$ and the unmatched letters form the word $z$.

\end{proof}

\begin{lemma}\label{simple_case_implication}

If $q\geq 1$, $m\geq 3$ and every $q$-ary word of length $(q-1)m^2-2qm+3m$ is in $S^{\dagger}$, then every $q$-ary word of length at least $qm^2-2qm+2m$ is in $S^{\dagger}$.

\end{lemma}

\begin{proof}

Assume that $y$ is a $q$-ary word of length at least $qm^2-2qm+2m$.
Since $|y|-[(q-1)m^2-2qm+3m]\geq
qm^2-2qm+2m-[(q-1)m^2-2qm+3m]=m^2-m>m^2-m-1=g(m, m+1)$, we can write
\[
|y|-[(q-1)m^2-2qm+3m] = rm + s(m+1)
\]
for some non-negative integers $r$ and $s$.  By
Lemma~\ref{enough_of_a_letter}, we can match letters in $y$ by making
$r$ choices from $\{x_1^m, x_2^m, \ldots, x_q^m \}$ and $s$ choices
from $\{x_1^{m+1}, x_2^{m+1},\ldots, x_q^{m+1}\}$, leaving exactly
$(q-1)m^2-2qm+3m$ unmatched letters. Let $z$ denote the word of
length $(q-1)m^2-2qm+3m$ consisting of these unmatched letters. By
hypothesis we have $z\in S^{\dagger}$.  Therefore $y$ consists of a
word in $S^{\dagger}$ shuffled with words in $S$ and so $y\in
S^{\dagger}$, as required.

\end{proof}

In the next lemma we will repeatedly apply the following procedure.

\parbox{\textwidth}{
$MATCH(y, i, j, r, \gamma)$:
\begin{itemize}
\item Repeat $r$ times for the first $r$ occurrences of $x_i$.
\begin{itemize}
\item Consider the first unmatched occurrence of $x_i$ in $y$.
\item If there are at least $m-1$ unmatched $x_j$'s to the left of the chosen $x_i$ then use a $x_j^{m-1}x_i$ to match $m-1$ of these $x_j$'s and this occurrence of $x_i$.
\item  If there are at least $m-1$ unmatched $x_j$'s to the right of the chosen $x_i$ then use a $x_ix_j^{m-1}$ to match $m-1$ of these $x_j$'s and this occurrence of $x_i$.
\end{itemize}
\item Use $\gamma$ $x_i^m$'s to match $\gamma m$ unmatched $x_i$'s.
\end{itemize}
}

\begin{lemma}\label{simple_case_every_word}

If $q\geq 2$ and $m\geq 3$, then every $q$-ary word of length $(q-1)m^2-2qm+3m$ is in $S^{\dagger}$.

\end{lemma}

\begin{proof}

By induction on $q$. First we prove the result for $q=2$. That is we prove that if $m\geq 3$ and $$S=\{x_1^m, x_1^{m+1}, x_2^m, x_2^{m+1}, x_1x_2^{m-1}, x_2^{m-1}x_1, x_2x_1^{m-1}, x_1^{m-1}x_2\},$$ then every binary word of length $m^2-m$ is in $S^{\dagger}$.

Let $y$ be a binary word of length $m^2-m$. If $|y|_{x_1}$ is a multiple of $m$, then clearly $|y|_{x_2}$ is also a multiple of $m$ and it follows that $y\in\{x_1^m, x_2^m\}^{\dagger}$. So assume $|y|_{x_1}=\gamma m+r$ where $\gamma$ and $r$ are non-negative integers such that $0\leq\gamma\leq m-2$ and $1\leq r\leq m-1$. It follows that $|y|_{x_2}=(m-2-\gamma)m+(m-r)$.

Case 1: $1\leq r\leq m-2-\gamma$.

Apply $MATCH(y, 1, 2, r, \gamma)$.
Lastly, use $m-r-1-\gamma$ $x_2^m$'s to match the remaining $x_2$'s. 
This process will match exactly $\gamma m+r$ $x_1$'s and $r(m-1)+(m-r-1-\gamma)m=m^2-m-r-\gamma m=(m-2-\gamma)m+(m-r)$ $x_2$'s. This procedure is possible because once we make it to the $r$-th $x_1$, 
there are at least $(m-2-\gamma)m+(m-r)-(r-1)(m-1)=m^2-\gamma m-rm-1\geq m^2-\gamma m-(m-2-\gamma)m-1=2m-1$ unmatched $x_2$'s ensuring the $r$-th $x_1$ is either preceded or followed by $m-1$ unmatched $x_2$'s.

Case 2: $m-\gamma\leq r\leq m-1$.

Apply $MATCH(y, 2, 1, m-r, m-2-\gamma)$.
Lastly, use $\gamma -m+r+1$ $x_1^m$'s to match the remaining $x_1$'s. 
This matches exactly $(m-2-\gamma)m+(m-r)$ $x_2$'s and $(\gamma-m+r+1)m+(m-r)(m-1)=\gamma m-m^2+rm+m+m^2-rm-m+r=\gamma m+r$ $x_1$'s. 
This procedure is possible because once we make it to the $(m-r)$-th $x_2$, 
there are at least $\gamma m+r-(m-r-1)(m-1)=\gamma m+r-m^2+m+rm-r+m-1=\gamma m-m^2+rm+2m-1\geq\gamma m-m^2+(m-\gamma)m+2m-1=2m-1$ unmatched $x_1$'s ensuring the $(m-r)$-th $x_2$ is either preceeded or followed by $m-1$ unmatched $x_1$'s.

Case 3: $r=m-1-\gamma$.

Subcase 1: $y$ ends with a $x_1$.

Apply $MATCH(y, 1, 2, m-\gamma-1, \gamma)$ but adjust the procedure so that once the first $m-\gamma-2$ $x_1$'s have been matched, match the last $x_1$ as opposed to the $(m-\gamma-1)$-st.
This procedure matches exactly $\gamma m+(m-\gamma-1)=\gamma m+r$ $x_1$'s and 
$(m-\gamma-1)(m-1)=m^2-\gamma m-2m+\gamma+1=(m-2-\gamma)m+(m-r)$ $x_2$'s. 
This procedure is possible because once we match the first $m-\gamma-2$ $x_1$'s, there will be $m-1$ unmatched $x_2$'s and they will all be followed by the last $x_1$.

Subcase 2: $y$ ends with a $x_2$.

Apply $MATCH(y, 2, 1, \gamma+1, m-2-\gamma)$ but adjust the procedure so that once the first $\gamma$ $x_2$'s have been matched, match the last $x_2$ as opposed to the $(\gamma+1)$-st.
This procedure matches exactly $(\gamma+1)(m-1)=\gamma m +m-1-\gamma=\gamma m+r$ $x_1$'s and $(m-2-\gamma)m+\gamma+1=(m-2-\gamma)m+(m-r)$ $x_2$'s. This procedure is possible because once we match the first $\gamma$ $x_2$'s, there will be $m-1$ unmatched $x_1$'s and they will all be followed by the last $x_2$.

In every case, all the letters in $y$ can be matched using words in $S$. Thus $y\in S^{\dagger}$ and the result follows for $q=2$.

Assume the result holds for some $q\geq2$. It suffices to show that the result holds for $q+1$. Let $y$ be a $(q+1)$-ary word of length $qm^2-2(q+1)m+3m$. Assume for the sake of contradiction that $|y|_{x_i}\geq m^2-2m+1$ for all $i\in\{1, 2, ..., q+1\}$. Then we get
\begin{eqnarray*}
|y|&\geq&(q+1)(m^2-2m+1)\\
&=&qm^2-2qm+q+m^2-2m+1\\
&=&qm^2-2(q+1)m +m^2+q-1\\
&\geq&qm^2-2(q+1)m+3m+q-1\\
&>&qm^2-2(q+1)m+3m.
\end{eqnarray*}
It follows that $|y|_{x_i}\leq m^2-2m$ for some letter $x_i$ and WLOG
we will say it is $x_a$. If $|y|_{x_a}=m^2-2m$, then we can match all
of the $x_a$'s with $m-2$ $x_a^m$'s. Thus, $y$ consists of a $q$-ary word of length $(q-1)m^2-2qm+3m$ shuffled with $m-2$ $x_a^m$'s. By our induction hypothesis, $y\in S^{\dagger}$ and we are done. So assume $|y|_{x_a}=\gamma m+r$ where $\gamma$ and $r$ are non-negative integers such that $0\leq\gamma\leq m-3$ and $0\leq r\leq m-1$. 
Let $l$ denote the number of letters in $y$ that are not $x_a$'s. 
Then $$l=qm^2-2(q+1)m+3m-\gamma m-r=qm^2-2qm+m-\gamma m -r.$$ 
By the pigeonhole principle, at least one of these letters (call it $x_b$) has at least 
$$\left\lceil{\frac{qm^2-2qm+m-\gamma m-r}{q}}\right\rceil=\left\lceil{m^2-2m+\frac{m-\gamma m-r}{q}}\right\rceil$$ occurrences in $y$. If $\gamma=0$, then $\left\lceil{m^2-2m+\frac{m-\gamma m-r}{q}}\right\rceil\geq m^2-2m+1=(m-1)^2$. If $\gamma\neq0$, then $\left\lceil{m^2-2m+\frac{m-\gamma m-r}{q}}\right\rceil\geq m^2-2m+m-\gamma m-r=m^2-m-\gamma m-r$. 

Case 1: $0\leq r\leq m-2-\gamma$.

Apply $MATCH(y, a, b, r, \gamma)$.
Once we get to the $r$-th $x_a$, there are at least $(m-1)^2-(r-1)(m-1)\geq(m-1)^2-(m-3)(m-1)=2m-2$ unmatched $x_b$'s if $\gamma=0$ and at least $m^2-m-\gamma m-r-(r-1)(m-1)=m^2-\gamma m-rm-1\geq m^2-\gamma m-(m-2-\gamma)m-1=2m-1$ unmatched $x_b$'s if $\gamma\neq0$. This ensures the $r$-th $x_a$ is either preceded or followed by $m-1$ unmatched $x_b$'s. 
This procedure leaves at least $qm^2-2(q+1)m+3m-(\gamma+r)m\geq
qm^2-2(q+1)m+3m-(\gamma+m-2-\gamma)m=(q-1)m^2-2qm+3m$ unmatched
letters.  These unmatched letters form a word over a $q$-letter
alphabet, which by Lemma~\ref{matching_multiples_of_m} can be written
as the shuffle of a word in $S^\dagger$ and a word $z$ of length
exactly $(q-1)m^2-2qm+3m$.  By the induction hypothesis, the $q$-ary
word $z$ is in $S^\dagger$.  It follows that $y$ is in $S^\dagger$ as well.

Case 2: $m-\gamma\leq r\leq m-1$.

Apply $MATCH(y, b, a, m-r, 0)$. Note that since $m^2-m-\gamma m-r\geq (m-1)m-(m-3)m-r\geq m-r$, there are at least $m-r$ $x_b$'s in $y$ ($\gamma=0$ is omitted since it does not fall under this case). Then use $\gamma-m+r+1$ $x_a^m$'s to match the remaining $x_a$'s.
Once we get to the $(m-r)$-th $x_b$, there will still be at least 
$\gamma m+r-(m-r-1)(m-1)=\gamma m +r-m^2+m+rm-r+m-1\geq\gamma m+(m-\gamma)m-m^2+2m-1=2m-1$ $x_a$'s ensuring that the $(m-r)$-th $x_b$ is either preceded or followed by $m-1$ $x_a$'s. 
This procedure leaves at least 
$qm^2-2(q+1)m+3m-[(m-r)m+(\gamma-m+r+1)m]=qm^2-2(q+1)m+3m-(\gamma+1)m\geq
qm^2-2(q+1)m+3m-(m-3+1)m=(q-1)m^2-2qm+3m$ unmatched letters.  As in
Case~1, we can use Lemma~\ref{matching_multiples_of_m} and the
induction hypothesis to show that the $q$-ary word consisting of the unmatched
letters is in $S^\dagger$, and consequently, so is $y$.

Case 3: $r=m-1-\gamma$.

Assume for the sake of contradiction that $|y|_{x_a}=\gamma m +(m-1-\gamma)$ and $|y|_{x_i}\geq m^2-2m+1$ for all $i$ such that $i\neq a$. Then 
\begin{eqnarray*}
|y|&\geq&\gamma m +(m-\gamma-1) +q(m^2-2m+1)\\
&=&\gamma(m-1)+m-1+qm^2-2qm+q\\
&\geq&qm^2-2qm+q+m-1\\
&=&qm^2-2(q+1)m+3m+q-1\\
&>&qm^2-2(q+1)m+3m.
\end{eqnarray*}
Thus, there exists another letter (call it $x_c$) such that $|y|_{x_c}\leq m^2-2m$. If $|y|_{x_c}=m^2-2m$, then you can match all of the $x_c$'s with $m-2$ $x_c^m$'s. Thus, $y$ consists of a $q$-ary word of length $(q-1)m^2-2qm+3m$ shuffled with $m-2$ $x_c^m$'s and we are done.
So assume $|y|_{x_c}=\delta m+s$ where $\delta$ and $s$ are non-negative integers such that 
$0\leq\delta\leq m-3$ and $0\leq s\leq m-1$. If $0\leq s\leq m-2-\delta$, then apply Case 1 except with $x_c$ rather than $x_a$. 
If $m-\delta\leq s\leq m-1$, then apply Case 2 except with $x_c$ rather than $x_a$. If not, then $s=m-1-\delta$. Assume that the last occurrence of an $x_c$ is after the last occurrence of an $x_a$ in $y$. Apply $MATCH(y, c, a, \gamma+1, 0)$ but adjust the procedure so that once the first $\gamma$ $x_c$'s have been matched, match the last $x_c$ as opposed to the $(\gamma+1)$-st.
Note that there are at least $\gamma+1$ $x_c$'s since $\gamma+1\geq m-3+1=m-2$ and there are at least $m-2$ $x_c$'s by definition. This procedure is possible because once we match the first $\gamma$ $x_c$'s, there will be $m-1$ unmatched $x_a$'s and they will all be followed by the last $x_c$.
This procedure leaves at least $qm^2-2(q+1)m+3m-(\gamma+1)m\geq
qm^2-2(q+1)m+3m-(m-3+1)m=(q-1)m^2-2qm+3m$ unmatched
letters. As in
Case~1, we can use Lemma~\ref{matching_multiples_of_m} and the
induction hypothesis to show that the $q$-ary word consisting of the unmatched
letters is in $S^\dagger$, and consequently, so is $y$.

The case is similar if the last occurrence of $x_a$ is after the last
occurrence of a $x_c$ in $y$ (just apply $MATCH(y, a, c, \delta+1, 0)$
but adjust the procedure so that once the first $\delta$ $x_a$'s have
been matched, match the last $x_a$ as opposed to the $(\delta+1)$-st).

In all cases we have shown that $y\in S^{\dagger}$, so the proof is complete.

\end{proof}

This sequence of lemmas completes the proof of
Theorem~\ref{simple_case_bound}.  Given that we know the length of a
longest word not in $S^\dagger$ for the particular family of sets $S$
studied in this section, a natural question would be to count exactly
how many words are not in $S^\dagger$.  Unfortunately, this seems to
be rather difficult.  Here is a rather weak lower bound.

\begin{theorem}

Let $t(S)$ denote the number of words not in $S^{\dagger}$. If $q\geq
2$ and $m\geq 3$,$$t(S) \geq \frac{q^{m^2+m+1}-q^{m^2+m}+q^{m+1}-q^{2m+1}+q^m-q}{(q-1)(q^m-1)(q^{m+1}-1)}.$$

\end{theorem}

\begin{proof}

Since each word in $S$ is of length $m$ or $m+1$, it follows that any
word with a length that cannot be written as a non-negative integer
linear combination of $m$ and $m+1$ is not in $S^{\dagger}$. It is
clear that $\ell$ cannot be written as a non-negative integer linear
combination of $m$ and $m+1$ if and only if $\ell=\gamma m+r$ for non-negative integers $\gamma$ and $r$ where $\gamma+1\leq r\leq m-1$. Therefore,
\begin{eqnarray*}
t&\geq&\sum_{i=1}^{m-1} q^i+\sum_{i=2}^{m-1} q^{m+i}+\cdots+\sum_{i=m-1}^{m-1} q^{(m-2)m+i}\\
&=&\sum_{j=1}^{m-1} \sum_{i=j}^{m-1} q^{(j-1)m+i}\\
&=&\sum_{j=1}^{m-1} \frac{q^{(j-1)m}(q^m-q^j)}{q-1}\\
&=&\frac{q^{m^2+m+1}-q^{m^2+m}+q^{m+1}-q^{2m+1}+q^m-q}{(q-1)(q^m-1)(q^{m+1}-1)}.
\end{eqnarray*}

\end{proof}

This lower bound could certainly be improved by more complicated
arguments, and similarly an upper bound could be calculated as well,
but it seems difficult to get an exact, closed-form expression for
$t(S)$.

\section{Conclusion}
One possibility for future work is to improve the upper bound given in
Theorem~\ref{upper bound}.  Another issue not addressed in this paper
concerns the computational complexity of the following problem:
\textit{Given a set of words $S$ such that $S^\dagger$ is co-finite,
  what is the length of a longest word not in $S^\dagger$?}  We have
given an upper bound for this length but we have not given an
algorithm to determine it exactly.  Returning for a moment to the
classical Frobenius problem, we note that Ram\'irez--Alfons\'in showed
that computing the Frobenius number for a given set of integers is
NP-hard with respect to Turing reductions \cite{RA96}.  We also
claimed in our introduction that the classical Frobenius problem is
equivalent to the special case of the problem considered in this paper
where the set $S$ is over a unary alphabet.  However, this is not
entirely true from the point of view of computational complexity.  In
the classical setting, the size of the inputs $m_1, m_2, \ldots, m_k$
would be measured in terms of the lengths of their base-$2$
representations (i.e., in terms of their base-$2$ logarithms); in our
setting, these integers are represented as the unary strings
$0^{m_1}, 0^{m_2}, \ldots, 0^{m_k}$.  With respect to this unary
representation, the classical Frobenius problem is solvable in
polynomial time.  Over larger alphabets, however, we don't know what
the computational complexity is of determining the length of a longest
word not in $S^\dagger$.

\end{document}